\documentclass[final]{dmtcs-episciences}

\usepackage[utf8]{inputenc}
\usepackage{subfigure}
\usepackage{amsmath,amssymb}

\usepackage[round]{natbib}

\newcommand{\ceil}[1]{\left\lceil{#1}\right\rceil}

\newenvironment{proofof}[1]{\medskip\noindent\emph{Proof of #1. }\ignorespaces}{\hfill$\Box$\medskip\par\noindent\ignorespacesafterend}

\newtheorem{corollary}{Corollary}
\newtheorem{lemma}{Lemma}
\newtheorem{obs}{Observation}
\newtheorem{proposition}{Proposition}
\newtheorem{theorem}{Theorem}

\author{Guillaume Ducoffe\affiliationmark{1,2}\thanks{This work was supported by a grant of Romanian Ministry of Research and Innovation CCCDI-UEFISCDI. project no. 17PCCDI/2018. The article also has enjoyed the support of the Romanian Young Academy, Stiftung Mercator and the
Alexander von Humboldt Foundation.}
  \and Michel Habib\affiliationmark{3}\thanks{Supported by Inria Gang project-team, and ANR project DISTANCIA (ANR-17-CE40-0015).}
  \and Laurent Viennot\affiliationmark{4}\thanks{Supported by Irif laboratory from CNRS and Paris University, and ANR project Multimod (ANR-17-CE22-0016).}}
\title{Fast Diameter Computation within Split Graphs}
\affiliation{
  University of Bucharest, Faculty of Mathematics and Computer Science, Romania\\
  National Institute for Research and Development in Informatics, Romania\\
  Paris University, and IRIF CNRS, France\\
  Inria, and Paris University, France}
\received{2020-04-23}
\revised{2021-05-24}
\accepted{2021-09-10}
\begin{document}
\publicationdetails{23}{2021}{3}{11}{6422}
\maketitle
\begin{abstract}
    {\em When can we compute the diameter of a graph in quasi linear time?}
    We address this question for the class of {\em split graphs}, that we observe to be the hardest instances for deciding whether the diameter is at most two.
    We stress that although the diameter of a non-complete split graph can only be either $2$ or $3$, under the Strong Exponential-Time Hypothesis (SETH) we cannot compute the diameter of an $n$-vertex $m$-edge split graph in less than quadratic time -- in the size $n+m$ of the input.
    Therefore it is worth to study the complexity of diameter computation on {\em subclasses} of split graphs, in order to better understand the complexity border.
    %
    Specifically, we consider the split graphs with bounded {\em clique-interval number} and their complements, with the former being a natural variation of the concept of interval number for split graphs that we introduce in this paper.
    We first discuss the relations between the clique-interval number and other graph invariants such as the classic interval number of graphs, the treewidth, the {\em VC-dimension} and the {\em stabbing number} of a related hypergraph.
    Then, in part based on these above relations, we almost completely settle the complexity of diameter computation on these subclasses of split graphs:

    \begin{itemize}
        \item For the $k$-clique-interval split graphs, we can compute their diameter in truly subquadratic time if $k={\cal O}(1)$, and even in quasi linear time if $k=o(\log{n})$ and in addition a corresponding ordering of the vertices in the clique is given. However, under SETH this cannot be done in truly subquadratic time for any $k = \omega(\log{n})$.
        \item For the {\em complements} of $k$-clique-interval split graphs, we can compute their diameter in truly subquadratic time if $k={\cal O}(1)$, and even in time ${\cal O}(km)$ if a corresponding ordering of the vertices in the stable set is given. Again this latter result is optimal under SETH up to polylogarithmic factors.
    \end{itemize}
    Our findings raise the question whether a $k$-clique interval ordering can always be computed in quasi linear time.
    We prove that it is the case for $k=1$ and for some subclasses such as bounded-treewidth split graphs, threshold graphs and comparability split graphs.
    Finally, we prove that some important subclasses of split graphs -- including the ones mentioned above -- have a bounded clique-interval number.
\end{abstract}

\section{Introduction}\label{sec:introduction}

In what follows, all graphs considered are assumed to be connected, unless stated otherwise.
Computing the diameter of a graph (maximum number of edges on a shortest path) is a fundamental problem with countless applications in computer science and beyond.
Unfortunately, the textbook algorithm for computing the diameter of an $n$-vertex $m$-edge graph takes ${\cal O}(nm)$-time.
This quadratic running-time is too prohibitive for large graphs with millions of nodes.
As already noticed in~\cite{ChD92,CDHP01}, an algorithm breaking this quadratic barrier for general graphs is unlikely to exist since it would lead to more efficient algorithms for some disjoint set problems and, as proved in~\cite{RoV13}, the latter would falsify the Strong Exponential-Time Hypothesis (SETH).
This raises the question of {\em when we can compute the diameter faster than ${\cal O}(nm)$.}
By restricting ourselves to more structured graph classes, here we hope in obtaining a finer-grained dichotomy between ``easy'' and ``hard'' instances for diameter computations -- with the former being quasi linear-time solvable ({\it i.e.}, solvable in ${\cal O}(m^{1+\epsilon})$ time for any $\epsilon > 0$) and the latter being impossible to solve in subquadratic time under some complexity assumptions. 

Specifically, we focus in this work on the class of split graphs, {\it i.e.}, the graphs that can be bipartitioned into a clique and a stable set.
-- For any undefined graph terminology, see~\cite{BoM08}. --
This is one of the most basic idealized models of core/periphery structure in complex networks~\cite{BoE00}.
We stress that every split graph has diameter at most three.
In particular, computing the diameter of a non-complete split graph boils down to decide whether this is either two or three.
Nevertheless, under the Strong Exponential-Time Hypothesis (SETH) the textbook algorithm is optimal even for split graphs~\cite{BCH16}.
We observe that the split graphs are in some sense the {\em hardest} instances for deciding whether the diameter is at most two.
For that, let us bipartition a split graph $G$ into a maximal clique $K$ and a stable set $S$; it takes linear time~\cite{Gol04}.
The sparse representation of $G$ is defined as $(K,\{N_G[v] \mid v \in S\})$\footnote{
This sparse representation is also called the neighbourhood set system of the stable set of $G$, see Sec.~\ref{sec:rel}.}.
-- Note that all our linear-time algorithms in this paper run in the size of this above representation. --

\begin{obs}\label{obs:hardness-split}
Deciding whether an $n$-vertex graph $G$ has diameter two can be reduced in linear time to deciding whether a $2n$-vertex split graph $G'$ (given by its sparse representation) has diameter two.
\end{obs}

\begin{proof}
If $G=(V,E)$ then, let $V'$ be a disjoint copy of $V$.
For every $v \in V$ we denote by $v' \in V'$ its copy.
We define $G'=(V \cup V',E')$ where $E' := \{ u'v' \mid u,v \in V \} \cup \{ uv' \mid v \in N_G[u] \}$.
By construction, $G'$ is a split graph with maximal clique $V'$ and stable set $V$.
Furthermore, $diam(G) \leq 2$ if and only if $diam(G') \leq 2$. 
\end{proof}

Since the conference version of this paper, it was also proved in~\cite{DuD21} that computing the diameter of chordal graphs ({\it i.e.}, graphs with no induced cycle of length $> 3$,  a far-reaching generalization of split graphs) could be reduced in randomized quasi linear time to computing the diameter of various split subgraphs given by their sparse representation. This result got used in~\cite{Duc20+} for faster diameter computation within chordal graphs of bounded asteroidal number.
We here address the fine-grained complexity of diameter computation on {\em subclasses} of split graphs.
By the above Observation~\ref{obs:hardness-split} and by~\cite{DuD21}, our results can be applied to the study of the diameter-two problem on general graphs and to the study of the diameter problem within chordal graphs.

\paragraph{Related work.}
Exact and approximate distance computations for chordal graphs have been a common research topic over the last few decades~\cite{BCD99,BDLL04,DoG02,Dra05}.
To the best of our knowledge, this is the first study on the complexity of diameter computation on split graphs.
However, there exist linear-time algorithms for computing the diameter on some other subclasses of chordal graphs such as: interval graphs~\cite{Ola90}, or strongly chordal graphs~\cite{BCD98}.
These results imply the existence of linear-time algorithms for diameter computations on interval split graphs and strongly chordal split graphs, among other subclasses.

Beyond chordal graphs, the complexity of diameter computation has been considered for many graph classes, {\it e.g.}, see~\cite{CDHP01} and the papers cited therein.
In particular, the diameter of general graphs with {\em treewidth} $o(\log{n})$ can be computed in quasi linear time~\cite{AVW16}, whereas under SETH we cannot compute the diameter of split graphs with clique-number (and so, treewidth) $\omega(\log{n})$ in subquadratic time~\cite{BCH16}.
We stress that our two first examples of ``easy'' subclasses, namely: interval split graphs and strongly chordal split graphs have unbounded treewidth.
Our work unifies almost all known tractable cases for diameter computation on split graphs -- and offers some new such cases -- through a new increasing hierarchy of subclasses.

Relatedly, we proved in a companion paper~\cite{DHV20} that on the proper minor-closed graph classes and the bounded-diameter graphs of constant distance VC-dimension (not necessarily split) we can compute the diameter in time ${\cal O}(mn^{1-\varepsilon}) = o(mn)$, for some small $\varepsilon > 0$.
We consider in this work some subclasses of split graphs of constant distance VC-dimension.
However, the time bounds obtained in~\cite{DHV20} are barely subquadratic.
For instance, although the graphs of constant treewidth fit in our framework, our techniques in~\cite{DHV20} do not suffice for computing their diameter in quasi linear time.
In fact, neither are they sufficient to explain why we can compute the diameter in subquadratic time on graphs of superconstant treewidth $o(\log{n})$.
Unlike~\cite{DHV20} this article is a new step toward characterizing the graph classes for which we can compute the diameter in quasi {\em linear time}.

\paragraph{Our results.}
We introduce a new invariant for split graphs, that we call the {\em clique-interval number}. 
Formally, for any $k \geq 1$, a split graph is {\em $k$-clique-interval} if there exists a total ordering $\tau$ of the vertices in the clique so that, for every vertex $v$ in the stable set, the neighbour set $N(v)$ consists of at most $k$ intervals of vertices in the ordering $\tau$.
The clique-interval number of a split graph is the minimum $k$ such that it is $k$-clique-interval.
Although this definition is quite similar to the one of the interval number~\cite{McG77}, we show in Sec.~\ref{sec:rel} that being $k$-clique-interval does not imply being $k$-interval, and vice-versa\footnote{
We observe that the {\em general graphs} with bounded interval number are sometimes called ``split interval''~\cite{BHNS+06}, that may create some confusion.}.
In fact, as we also prove in Sec.~\ref{sec:rel}, the clique-interval number of a split graph is more closely related to the {\em VC-dimension} and the {\em stabbing number} of the neighbourhood set system of its stable set.
Nevertheless, a weak relationship with the interval number can also be derived in this way.

We then study in Sec.~\ref{sec:range-queries} what the complexity of computing the diameter is on split graphs parameterized by the clique-interval number.
It follows from our results in Sec.~\ref{sec:rel} and those in~\cite{DHV20} that on every subclass of {\em constant} clique-interval number, there exists a subquadratic-time algorithm for diameter computation.
Our work completes this general result as, if we assume a total ordering over the clique to be given in the input (showing the input split graph to be $k$-clique-interval for some value $k$), it provides an almost complete characterization of the quasi linear-time solvable instances. 
As a warm-up, we observe in Sec.~\ref{sec:univ-vertex} that on {\em clique-interval} split graphs ({\it a.k.a.}, $1$-clique-interval), deciding whether the diameter is two is equivalent to testing for a universal vertex.
We give a direct proof of this result and another one based on the inclusion of clique-interval split graphs in the subclass of {\em strongly chordal} split graphs.
On the way, we prove more generally -- and perhaps surprisingly -- that for the intersection of split graphs with many interesting graph classes from the literature, having diameter at most two is equivalent to having a universal vertex!   
Then, we address the more general case of $k$-clique-interval split graphs, for $k \geq 2$.
\begin{itemize}
\item Our first main contribution is the following almost dichotomy result (Theorem~\ref{thm:main}).
For every $n$-vertex $k$-clique-interval split graph, we can compute its diameter in quasi linear-time if $k = o(\log{n})$ and a corresponding total ordering of its clique is given. 
This result follows from an all new application of a generic framework based on {\em $k$-range trees}~\cite{Ben79}, that was already used for diameter computations on some special cases~\cite{AVW16,Duc19} but with a quite different approach than ours.
Furthermore, the logarithmic upper bound on $k$ is somewhat tight -- at least if we assume every $k$-clique-interval split graph to be given with a $k$-clique-ordering.
Indeed, we also prove that under SETH  we cannot compute in subquadratic time the diameter of $k$-clique-interval split graphs for $k = \omega(\log{n})$. The complexity of computing the diameter of a $k$-clique-interval split graph, for $k = o(\log{n})$, remains open if we are {\em not} given any total order over the vertices of the clique.

\smallskip
We note that this above result is quite similar to the one obtained in~\cite{AVW16} for treewidth.
Indeed, we observe that every split graph of treewidth $k$ is $k$-clique-interval (and even $\left\lceil \frac k 2 \right\rceil$-clique-interval, see Sec.~\ref{sec:rel}).
We use this easy observation so as to prove our conditional time complexity lower-bound.

\smallskip
\item Then, we focus on the {\em complements} of $k$-clique-interval split graphs --- these are an interesting subclass in their own right since they generalize, {\it e.g.}, interval split graphs. For the latter we get a more straightforward algorithm for diameter computation, with a better dependency in $k$. Indeed, we prove that we can compute the diameter of such graphs in time ${\cal O}(km)$ if a corresponding ordering is given. This result is conditionally optimal up to polylogarithmic factors (at least if we assume every such split graph to be given with a $k$-clique-ordering of its complement) because $k = {\cal O}(n)$ and, under SETH, we cannot compute the diameter of split graphs with $m = \tilde{\cal O}(n)$ edges in subquadratic time~\cite{AVW16}. The complexity of computing the diameter within complements of $k$-clique-interval split graphs remains open if we are not given any total ordering over the vertices in the stable set.
\end{itemize}

It follows from these two above results that having at hands a $k$-clique-interval ordering for a split graph or its complement can help to significantly improve the time complexity for computing its diameter.
We so ask whether such orderings can be computed in quasi linear time, for some small values of $k$.
In Sec.~\ref{sec:recognition} we prove that it is indeed the case for bounded-treewidth split graphs (trivially) and some other dense subclasses of bounded clique-interval number such as comparability split graphs.
Finally our main result in this section is that the clique-interval split graphs can be recognized in linear time.

\medskip
Overall, we believe that our study of $k$-clique-interval split graphs is a promising framework in order to prove, using this parameter, new quasi linear-time solvable special cases for diameter computations on split graphs and beyond.

\smallskip
Results of this paper were partially presented at the COCOA'19 conference~\cite{DHV19}.

\section{Clique-interval numbers and  other graph parameters}\label{sec:rel}

We start by relating the clique-interval number of split graphs with better-studied invariants from Graph Theory and Computational Geometry.

\paragraph{Treewidth.}
First we observe that if $G=(S\cup K,E)$ is a split graph then, for {\em any} total order over $K$ and any $v \in S$, $N_G(v)$ consists of at most $|K|$ intervals of vertices. 
In fact we can improve this rough upper-bound, as follows:

\begin{lemma}\label{lem:bounded-tw}
Every split graph with clique-number (and so, treewidth) at most $k$ is $\ceil{\frac{k}{2}}$-clique-interval.
\end{lemma}

\begin{proof}
Let $G=(K \cup S,E)$ be a split graph with clique-number $|K| = k$.
Let $\tau$ be any total ordering of $K$. For any integer $i\le |K|$, and any $v \in S$, we claim that $N_G(v)$ is the union of at most $\max\{k-i, i\}$ intervals of $\tau$.
Indeed if $v$ has at least $i$ non-neighbours, then it has at most $k-i$ neighbours and so the claim trivially holds.
Otherwise, $v$ has less than $i$ non-neighbours in $K$, but then the other nodes form at most $i$ intervals of $\tau$.
This bound is optimal for $i = \ceil{\frac k 2}$.
\end{proof}

Conversely, since any complete graph is $1$-clique-interval, the clique-interval number cannot be bounded from below by any function of the treewidth.

\paragraph{VC-dimension and Stabbing number.}
For a set system $(X,R)$ ({\it a.k.a.}, range space or hypergraph), we say that a subset $Y \subseteq X$ is {\em shattered} if $\{ Y \cap r \mid r \in R \}$ is the power-set of $Y$. The {\em VC-dimension} of a finite set system is the largest size of a shattered subset.
We now prove an interesting connection between the clique-interval number of a split graph and the VC-dimension of a related set system.

\begin{proposition}\label{prop:vc-dim}
For any split graph $G=(K \cup S, E)$, let ${\cal S} = \{ N_G(u) \mid u \in S \}$.
If $G$ is $k$-clique-interval then $(K,{\cal S})$ has VC-dimension at most $2k$.
\end{proposition}

\begin{proof}
Suppose by contradiction that $G$ is $k$-clique-interval but $(K,{\cal S})$ has VC-dimension at least $2k+1$.
Let $Y \subseteq K$ be such that $|Y| = 2k+1$ and $Y$ is shattered.
Any total ordering of $K$ induces a total ordering $y_1,y_2,\ldots,y_{2k+1}$ over $Y$.
Since $Y$ is shattered there is an $u \in S$ such that $N_G(u) \cap Y = \{ y_{2i+1} \mid 0 \leq i \leq k \}$.
But then, $N_G(u)$ is the union of at least $k+1$ intervals, a contradiction.
\end{proof}

It turns our that a weak converse of Proposition~\ref{prop:vc-dim} also holds.
We need a bit of terminology from~\cite{ChW89}.
A {\em spanning path} for $(X,R)$ is a total ordering $x_1,x_2,\ldots,x_{|X|}$ of $X$.
Its {\em stabbing number} is the maximum number of consecutive pairs $x_i,x_{i+1}$ that a set $r \in R$ can {\em stab}, {\it i.e.}, for which $|r \cap \{x_i,x_{i+1}\}| = 1$. 
Finally, the stabbing number of $(X,R)$ is the minimum stabbing number over its spanning paths.

\begin{obs}\label{obs:stabbing-num}
For any split graph $G=(K \cup S, E)$, let ${\cal S} = \{ N_G(u) \mid u \in S \}$.
The clique-interval number of $G$ is within one of the stabbing number of $(K,{\cal S})$.
\end{obs}

The main result of~\cite{ChW89} is that every range system of VC-dimension at most $k$ has a stabbing number in $\tilde{\cal O}(f(k) \cdot n^{1-\frac{1}{f(k)}})$, for some exponential function $f$.
We so obtain:

\begin{corollary}\label{cor:vc-dim}
For any split graph $G=(K \cup S, E)$, let ${\cal S} = \{ N_G(u) \mid u \in S \}$.
If $(K,{\cal S})$ has VC-dimension at most $k$ then $G$ is $\tilde{\cal O}(f(k) \cdot n^{1-\frac{1}{f(k)}})$-clique-interval, for some exponential function $f$.
\end{corollary}

\paragraph{Interval number.}
Finally, we relate the clique-interval number of split graphs with their interval number.
A graph $G=(V,E)$ is called $k$-interval if we can map every $v \in V$ to the union of at most $k$ closed interval on the real line, denoted by $I(v)$, in such a way that $uv \in E \Longleftrightarrow I(u) \cap I(v) \neq \emptyset$. In particular, $1$-interval graphs are exactly the interval graphs. 

\begin{figure}
\centering
\includegraphics[width=.25\textwidth]{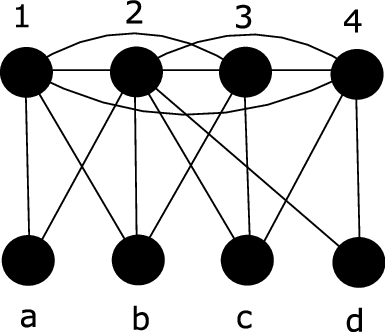}
\caption{An interval split graph that is not a clique-interval split graph.}
\label{fig:interval-not-clique-interval}
\end{figure}

We observe that there are $k$-interval split graphs that are {\em not} $k$-clique-interval, already for $k=1$.
For instance, consider the interval split graph of Fig.~\ref{fig:interval-not-clique-interval} (note that $\{a,1,2\},\{b,1,2,3\},\{1,2,3,4\},\{c,2,3,4\},\{d,2,4\}$ is a linear ordering of its maximal cliques).
Suppose by contradiction it is a clique-interval split graph, and let us consider a corresponding total ordering of $K=\{1,2,3,4\}$.
The intersection $N(b) \cap N(c) = \{2,3\}$ should be an interval.
But then, one of the pairs $\{1,2\}$ or $\{4,2\}$ is not consecutive, and so one of $N(a)$ or $N(d)$ is not an interval.
Therefore, the graph of Fig.~\ref{fig:interval-not-clique-interval} is not a clique-interval split graph.
Conversely, there are clique-interval split graphs which are not interval graphs. 
For instance, this is the case of thin spiders, {\it i.e.}, split graphs such that the edges between the maximal clique and the stable set induce a perfect matching.

Nevertheless we prove a weak connection between the clique-interval number and the interval number of a split graph, by using the VC-dimension.
Specifically, Bousquet et al. proved that the neighbourhood set system of any interval graph has VC-dimension at most two~\cite{BLLP+15}.
We generalize their result.
For that, we use in our proof the well-known Sauer-Shelah-Perles' Lemma~\cite{Sau72,She72}, namely: for any set system $(X,R)$ and any subset $Y \subseteq X$, if the VC-dimension of $(X,R)$ is at most $d$, then the cardinality of $\{ Y \cap r \mid r \in R \}$ is an ${\cal O}(|Y|^d)$.

\begin{proposition}\label{prop:k-interval}
The neighbourhood set system of any $k$-interval graph has VC-dimension at most $(4+o(1))k\log{k}$.
\end{proposition}

\begin{proof}
Let $G=(V,E)$ be a $k$-interval graph, and let us fix a corresponding $k$-interval representation.
We create a graph $H$ whose vertices are the intervals in this representation and such that there is an edge between every two intersecting intervals.
Observe that $H$ is an interval graph, and so, its neighbourhood system has VC-dimension at most two~\cite{BLLP+15}.
Now, let $Y \subseteq V$, $|Y| = d$, be shattered by the neighbourhood set system of $G$.
For the corresponding interval set $I(Y) = \bigcup_{v \in Y} I(v)$, the cardinality of $\{ N_H(\alpha) \cap I(Y) \mid \alpha \in V(H) \}$ is an ${\cal O}(|I(Y)|^2) = {\cal O}((kd)^2)$ by the Sauer-Shelah-Perles' Lemma~\cite{Sau72,She72}.
But then, for any $u \in V$, there are only ${\cal O}((kd)^{2k})$ possibilities for $\Phi_Y(u) = \{ N_H(\beta) \cap I(Y) \mid \beta \in I(u) \}$.
Since $Y$ is shattered and, for any $u,v \in V$, $N(u) \cap Y \neq N(v) \cap Y \Longrightarrow \Phi_Y(u) \neq \Phi_Y(v)$, we so obtain $2^d = \left| \{ Z \mid \exists u \in V, \ N(u) \cap Y = Z \} \right| \leq \left|\{ {\cal J} \mid \exists u \in V, \ \Phi_Y(u) = {\cal J} \}\right| = {\cal O}((kd)^{2k})$.
As a result, we must have that $d \leq \log{{\cal O}((kd)^{2k})} = 2k\log{(kd)} + {\cal O}(1) = 2k\log{d} + 2k\log{k} + {\cal O}(1)$. Assume $d \geq k$ (else, we are done). 
Then, $2k\log{d} + 2k\log{k} + {\cal O}(1) \leq 4k\log{d} + {\cal O}(1) = (4+o(1))k\log{d}$. In particular, $d/\log{d} \leq (4+o(1))k$, that implies $d \leq (4+o(1))k\log{k}$.
\end{proof}

Using previous Corollary~\ref{cor:vc-dim}, we have:

\begin{corollary}\label{cor:k-interval}
If $G$ is a $k$-interval split graph then, $G$ is also $\tilde{\cal O}(f(k\log{k}) \cdot n^{1-\frac{1}{f(k\log{k})}})$-clique-interval, for some exponential function $f$.
\end{corollary}

\section{Diameter Computation in quasi linear time}\label{sec:range-queries}

We now address the time complexity of diameter computation on $k$-clique-interval split graphs.
Our first result in this section follows from the relations proved in Sec.~\ref{sec:rel} with the VC-dimension.

\begin{theorem}[~\cite{DHV20}]\label{thm:diam2-vc}
For every $d > 0$, there exists a constant $\varepsilon_d \in (0;1)$ such that in time $\tilde{\cal O}(mn^{1-\varepsilon_d})$ we can decide whether a graph whose neighbourhood set system has VC-dimension at most $d$ has diameter two.
\end{theorem}

\begin{corollary}\label{cor:diam2-vc}
For every constant $k$, there exists a constant $\eta_k \in (0;1)$ such that in time $\tilde{\cal O}(mn^{1-\eta_k})$ we can decide whether a $k$-clique-interval split graph has diameter two.
\end{corollary}

\begin{proof}
Let $G=(K \cup S,E)$ be a $k$-clique-interval split graph.
By Theorem~\ref{thm:diam2-vc}, it suffices to prove that the neighbourhood set system of $G$ has VC-dimension bounded by a function of $k$.
The latter system is the union of $(K,\{N_G(v) \mid v \in S \})$ with $(S\cup K,\{N_G(v) \mid v \in K\})$.
Since the VC-dimension of the union of two set systems can be upper bounded by a function of their respective VC-dimensions~\cite{Dud78}, we are left proving that both set systems have a VC-dimension which is upper bounded by a function of $k$.  
By Proposition~\ref{prop:vc-dim}, $(K,\{N_G(v) \mid v \in S \})$ has VC-dimension at most $2k$.
Furthermore, we claim that the VC-dimension of $(S\cup K,\{N_G(v) \mid v \in K\})$ is at least the VC-dimension of $(S,\{N_G(v) \cap S \mid v \in K\})$, and at most this value plus one. The lower bound is trivial because any subset shattered by $(S,\{N_G(v) \cap S \mid v \in K\})$ is also shattered by $(S\cup K,\{N_G(v) \mid v \in K\})$. Conversely, let $Y \subseteq K \cup S$ be shattered by $(S\cup K,\{N_G(v) \mid v \in K\})$. If $Y \subseteq S$ then, it is also shattered by $(S,\{N_G(v) \cap S \mid v \in K\})$, and so we are done. Therefore, let $K_Y = K \cap Y$ be nonempty. Since we have $K \subseteq N_G[v]$ for every $v \in K$, it follows that every vertex $v \in K$ is adjacent to at least $|K_Y|-1$ vertices of $K_Y$ (and even to all vertices of $K_Y$ if and only if $v \notin K_Y$). But then, since we assume $K_Y \subseteq Y$ to be shattered, we get $|K_Y| \leq 1$. Therefore, the claim is proved, and so we are left upper bounding the VC-dimension of $(S,\{N_G(v) \cap S \mid v \in K\})$.
We observe that $(K,\{N_G(v) \mid v \in S \})$ and $(S,\{N_G(v) \cap S \mid v \in K\})$ are dual from each other (where the dual of a set system $(X,R)$ is another set system $(R,X^*)$ so that for every $x \in X$, there is an $r_x \in X^*$ such that $r_x = \{r \in R \mid x \in r\}$).
Therefore, if the VC-dimension of $(K,\{N_G(v) \mid v \in S \})$ is at most some value $d$, then the VC-dimension of $(S,\{N_G(v) \cap S \mid v \in K\})$ is at most $2^d$~\cite[Lemma 2.3]{ChW89}.
We conclude from Proposition~\ref{prop:vc-dim} that the VC-dimension of $(S,\{N_G(v) \cap S \mid v \in K\})$ is at most $4^k$, and we are done using Theorem \ref{thm:diam2-vc}.
\end{proof}

We observe that Corollary~\ref{cor:diam2-vc} also holds for the {\em complements} of $k$-clique-interval split graphs -- with essentially the same proof as above.
In the remainder of this section we focus on the existence of {\em quasi linear-time} algorithms.
It starts with a small digression on clique-interval split graphs. 

\subsection{The Case $k=1$ and beyond}\label{sec:univ-vertex}

\begin{proposition}\label{prop:clique-interval}
A clique-interval split graph has diameter at most two if and only if it has a universal vertex.
Therefore, we can compute the diameter of clique-interval split graphs in linear time.
\end{proposition}

\begin{proof}
For $G=(K \cup S,E)$ a clique-interval split graph, let us fix a corresponding total order over the maximal clique $K$.
Since $K$ is a dominating clique of $G$, all vertices in $K$ are at a distance at most two from any other vertex in $G$.
Therefore, we are left to decide whether every two vertices in the stable set $S$ are at distance two, {\it i.e.}, whether they have a common neighbour.
Equivalently, given the total order over $K$, we must decide whether for every $v,v' \in S$, the intervals $N_G(v)$ and $N_G(v')$ intersect.
By the Helly property, a finite family of intervals on the line pairwise intersect if and only if they have a non-empty intersection.
As a result, we are left deciding whether $\bigcap_{v \in S} N_G(v) \neq \emptyset$, or equivalently whether some vertex of $K$ is adjacent to all of $S$.
Since $K$ is a clique, having a vertex adjacent to all the vertices in $S$ is equivalent to having a universal vertex. 
\end{proof}

We think that our proof of Prop.~\ref{prop:clique-interval} is a nice introduction to the properties of clique-interval representations for split graphs, that will be further exploited in Sec.~\ref{sec:gal-case}. 
In the remainder of this part we give an alternative proof of Prop.~\ref{prop:clique-interval} that also applies to larger subclasses of split graphs.
It starts with the following inclusion lemma.

\begin{lemma}\label{lem:strongly-chordal}
Every clique-interval split graph is strongly chordal.
\end{lemma}

\begin{proof}
The $n$-sun graph is a split graph $G_n = (K \cup S,E)$ such that $K = \{u_0,u_1,\ldots,u_{n-1}\}, \ S = \{v_0,v_1,\ldots,v_{n-1} \}$ and for every $i$ we have $N_G(v_i) = \{u_i,u_{i+1}\}$ (indices are taken modulo $n$).
A chordal graph is strongly chordal if and only if it does not contain any $n$-sun graph as an induced subgraph, for every $n \geq 3$~\cite{Far83}.
Since every split graph is a chordal graph, in order to prove the lemma, it suffices to prove that a clique-interval split graph does not contain any $n$-sun graph as an induced subgraph, for every $n \geq 3$.
This is always the case for clique-interval split graphs because, in any $n$-sun graph, the neighbourhoods of the vertices in the stable set induce a {\em cyclic} ordering over $K$.
\end{proof}

We stress that since every $n$-sun graph is $2$-clique-interval, this above Lemma~\ref{lem:strongly-chordal} does not hold for $k$-clique-interval graphs, for any $k \geq 2$.
Now, a {\em maximum neighbour} of $v$ is any $u \in V$ such that $\bigcup_{w \in N_G[v]} N_G[w] \subseteq N_G[u]$ (see~\cite{BCD98}).
We prove that if at least one vertex in a graph has a maximum neighbour then deciding whether the diameter is at most two becomes a trivial task.
In particular, this is always the case for strongly chordal graphs~\cite{Far83}.

\begin{lemma}\label{lem:max-neighbour}
For every $G=(V,E)$ and $u,v \in V$ such that $u$ is a maximum neighbour of $v$, we have $diam(G) \leq 2$ if and only if $u$ is a universal vertex.
\end{lemma}

\begin{proof}
If $u$ is universal in $G$ then, trivially, $diam(G) \leq 2$.
Conversely, assume $diam(G) \leq 2$.
Then, for every $x \in V$, $dist_G(x,v) \leq 2$, and so there exists a $w \in N_G[v]$ such that either $x=w$ or $xw \in E$.
In particular, $x \in N_G[w] \subseteq N_G[u]$.
As a result, we have $N_G[u] = V$ or, equivalently, $u$ is universal.
\end{proof}

Prop.~\ref{prop:clique-interval} now follows from the combination of Lemmata~\ref{lem:strongly-chordal} and~\ref{lem:max-neighbour}.
Furthermore it turns out that many well-structured graph classes ensure the existence of a vertex with a maximum neighbour such as: graphs with a pendant vertex, threshold graphs~\cite{HeK07} and interval graphs~\cite{BHMW10}, or even more generally dually chordal graphs~\cite{BCD98}.
We conclude that:

\begin{corollary}\label{cor:universal-vertex}
We can compute the diameter of split graphs with minimum degree one and dually chordal split graphs in linear time.

In particular, we can compute the diameter of interval split graphs and strongly chordal split graphs in linear time. 
\end{corollary}

We observe that we can easily modify the hardness reduction from~\cite{BCH16} in order to show that, under SETH, we cannot compute in subquadratic time the diameter of split graphs with minimum degree two.
Indeed, this aforementioned reduction outputs a split graph $G=(K \cup S,E)$ such that: $S$ is partitioned in two disjoint subsets $A$ and $B$; and a diametral pair must have one end in each subset.
Let $a,b,v_0 \notin V$ be fresh new vertices, and let $G' = (K' \cup S',E')$ be such that: $K' = K \cup \{a,b\}, \ S' = S \cup \{v_0\}$ and:
$$E' = E \cup \{ au \mid u \in A \} \cup \{ bw \mid w \in B\} \cup \{av,bv \mid v \in K \cup \{v_0\}\} \cup \{ab\}.$$
By construction, $diam(G') \leq 2$ if and only if $diam(G) \leq 2$.
Therefore, Corollary~\ref{cor:universal-vertex} is optimal for the parameter minimum degree.

\subsection{The general case}\label{sec:gal-case}

We are now ready to prove the first main result of this paper.

\begin{theorem}\label{thm:main}
If $G=(K \cup S,E)$ is an $n$-vertex split graph and we are given a total ordering over $K$ showing that $G$ is $k$-clique-interval then, we can compute the diameter of  $G$ in time ${\cal O}(m + k^22^{{\cal O}(k)}n^{1+\epsilon})$, for any $\epsilon > 0$.
This is quasi linear-time if $k=o(\log{n})$.

Conversely, under SETH we cannot compute the diameter of $n$-vertex split graphs with clique-interval number $\omega(\log{n})$ in subquadratic time.
\end{theorem}

In order to prove Theorem~\ref{thm:main}, we will use a special instance of {\em $k$-range tree}.
Such a data-structure stores a static set of $k$-dimensional points and it supports the following operation:
\begin{itemize}
\item ({\it Range Query}) Given $k$ intervals $[l_i;u_i], \ 1 \leq i \leq k$, compute the number of stored points $p$ such that, for every $1 \leq i \leq k$, we have $l_i \leq p_i \leq u_i$\footnote{
We refer to~\cite{BHM18} for a more general presentation of range trees.}. 
\end{itemize}

\begin{proposition}[\cite{BHM18}]\label{prop:complexity-range-trees}
%
For any $k \geq 2$ and any $n$-set of $k$-dimensional points, we can construct a $k$-range tree in time ${\cal O}(k\binom{k + 1 + \left\lceil \log{n} \right\rceil}{k+1}n) = {\cal O}(k2^{{\cal O}(k)}n^{1+\epsilon})$ for any $\epsilon > 0$, and answer any range query in time ${\cal O}(2^k\binom{k + 1+ \left\lceil \log{n} \right\rceil}{k+1}) = {\cal O}(2^{{\cal O}(k)}n^{\epsilon})$.
\end{proposition}

The use of range queries for diameter computation dates back from~\cite{AVW16} (see also~\cite{BHM18,Duc19} for some further applications).
Roughly, if in a graph $G$ we can find a separator $S$ of size at most $k$ that disconnects a diametral pair of $G$, then the idea is to compute a ``distance profile'' for every vertex $v \notin S$ w.r.t. $S$, and to see this profile as a $k$-dimensional point.
We can compute a diametral pair by constructing a $k$-range tree for these points and computing ${\cal O}(kn)$ range queries.
Our present approach is different from the one in~\cite{AVW16} as we define our multi-dimensional points based on some interval representation of the graph rather than on distance profiles, and we use a different type of range query than in~\cite{AVW16}. 

\begin{proofof}{Theorem~\ref{thm:main}}
Let $G=(K \cup S,E)$ be an $n$-vertex split graph, and let us assume to be given a total ordering over $K$ showing that $G$ is $k$-clique-interval.
By the hypothesis for every $v \in S$, we have $N_G(v) = \bigcup_{i=1}^{k} [l_i(v);u_i(v)]$ is the union of $k$ intervals such that $l_1(v) \leq u_1(v) < l_2(v) \leq u_2(v) < \ldots < l_{k}(v) \leq u_{k}(v)$.
Furthermore since the ordering over $K$ is given, the $2k$ endpoints that delimit these intervals can be computed in time ${\cal O}(|N_G(v)|)$, simply by scanning the neighbours of vertex $v$. 
We so map every vertex $v \in S$ to the $2k$-dimensional point $p(v) = (l_1(v),u_1(v),\ldots,l_{k}(v),u_{k}(v))$, that takes total time ${\cal O}(m)$. 
Then, we construct a $2k$-range tree for the points $p(v), v \in S$, that takes time ${\cal O}(k2^{{\cal O}(k)}n^{1+\epsilon})$ for any $\epsilon > 0$ by Proposition~\ref{prop:complexity-range-trees}.

For every $v \in S$, we are left with computing the number of vertices in $S$ at distance two from $v$.
Indeed, $G$ has diameter at most two if and only if this number is $|S|-1$ for every $v \in S$.
For that, we first observe that, for every $w \in S \setminus \{v\}$, we have $dist_G(v,w) = 2$ if and only if there exists an interval $[l_j(w);u_j(w)]$, for some $1 \leq j \leq k$, which contains a vertex of $N_G(v)$. We partition the vertices to count into disjoint subsets $S_1,S_2,\ldots,S_k$ such that, for every $1 \leq j \leq k$, $S_j$ contains all the vertices $w \in S$ which satisfy that: $[l_j(w);u_j(w)]$ contains a vertex of $N_G(v)$, and there is no interval $[l_{j'}(w);u_{j'}(w)]$ with this property for any $1 \leq j' < j$. Then, we further observe that for an interval $[l_j(w);u_j(w)]$ to contain a vertex of $N_G(v)$, it is necessary and sufficient for this interval to intersect $[l_i(v);u_i(v)]$, for some $1 \leq i \leq k$. Therefore, we sub-partition each subset $S_j$ into $S_{1,j}, S_{2,j}, \ldots, S_{k,j}$ such that, for every $1 \leq i \leq k$ and for every $w \in S_{i,j}$, we have: $[l_j(w);u_j(w)]$ and $[l_i(v);u_i(v)]$ intersect, and $[l_j(w);u_j(w)]$ and $[l_{i'}(v);u_{i'}(v)]$ do not intersect for any $1 \leq i' < i$.

Let $w \in S \setminus \{v\}$ be arbitrary and let $1 \leq i,j \leq k$ be fixed. We claim that $w \in S_{i,j}$ if and only if the following three constraints are satisfied:
\begin{enumerate}
\item $[l_i(v);u_i(v)] \cap [l_j(w);u_j(w)] \neq \emptyset$;
\item while for every $1 \leq i' < i$, $[l_{i'}(v);u_{i'}(v)] \cap [l_j(w);u_j(w)] = \emptyset$;
\item and for every $1 \leq i' \leq i$ and $1 \leq j' < j$, $[l_{i'}(v);u_{i'}(v)] \cap [l_{j'}(w);u_{j'}(w)] = \emptyset$.
\end{enumerate}
Indeed, if $w \in S_{i,j}$ then the first and second constraints hold by minimality of index $i$ (for fixed $j$) and the third constraint holds by minimality of index $j$. Conversely, assume the three constraints to hold. The first and second constraints imply that index $i$ the least index $i'$ such that the intervals $[l_j(w);u_j(w)]$ and $[l_{i'}(v);u_{i'}(v)]$ intersect. Therefore, in order to prove that $w \in S_{i,j}$, it suffices to prove that there is no interval $[l_{j'}(w);u_{j'}(w)]$, for any $1 \leq j' < j$, which contains a vertex of $N_G(v)$. Suppose for the sake of contradiction the existence of an interval $[l_{j'}(w);u_{j'}(w)]$, for some $1 \leq j' < j$, which contains a vertex of $N_G(v)$. Let $1 \leq i' \leq k$ be such that $[l_{j'}(w);u_{j'}(w)]$ and $[l_{i'}(v);u_{i'}(v)]$ intersect. The third constraint implies that $i' > i$. But then, since we also have that $[l_j(w);u_j(w)]$ and $[l_i(v);u_i(v)]$ intersect (first constraint), we get $l_{i'}(v) > l_j(w) > u_{j'}(w)$, and therefore $[l_{i'}(v);u_{i'}(v)] \cap [l_{j'}(w);u_{j'}(w)] = \emptyset$, a contradiction.

Next, we define range queries such that, for every $w \in S_{i,j}$, its point $p(w)$ is counted by exactly one of these queries. The first and second constraints are equivalent to one of the following three {\em disjoint} possibilities:
\begin{itemize}
\item $u_{i-1}(v) < l_j(w) \leq l_i(v)$ and $u_i(v) \leq u_j(w)$ (then, we set $[l_{2j-1};u_{2j-1}] = (u_{i-1}(v);l_i(v)]$ and $[l_{2j};u_{2j}]=[u_i(v);+\infty)$ for, respectively the $(2j-1)^{\text{th}}$ and $2j^{\text{th}}$ intervals of the range query);
\item or $u_{i-1}(v) < l_j(w) \leq l_i(v)$ and $l_i(v) \leq u_j(w) < u_i(v)$ (then, we set $[l_{2j-1};u_{2j-1}] = (u_{i-1}(v);l_i(v)]$ and $[l_{2j};u_{2j}]=[l_i(v);u_i(v))$ for, respectively the $(2j-1)^{\text{th}}$ and $2j^{\text{th}}$ intervals of the range query);
\item or $l_i(v) < l_j(w) \leq u_i(v)$ (then, we set $[l_{2j-1};u_{2j-1}] = (l_i(v);u_i(v)]$ for the $(2j-1)^{\text{th}}$ interval of the range query).
\end{itemize}
The third constraint is equivalent to have $\bigcup_{j' < j} [l_{j'}(w);u_{j'}(w)] \subseteq (-\infty,l_1(v)) \cup \left( \bigcup_{i' < i} (u_{i'}(v),l_{i'+1}(v)) \right)$.
Note that in order to subdivide this constraint into disjoint possibilities, it suffices to indicate: \texttt{(i)} the subset of all intervals among $(-\infty,l_1(v)) \cup \left( \bigcup_{i' < i} (u_{i'}(v),l_{i'+1}(v)) \right)$ that contain an interval $[l_{j'}(w);u_{j'}(w)]$ for some $j' < j$;
and (ii) the set of all indices $j' \in (1;j)$  such that $[l_{j'-1}(w);u_{j'-1}(w)]$ and $[l_{j'}(w);u_{j'}(w)]$ are {\em not} contained in the same such interval.
Note also that, for \texttt{(i)}, there are $2^i$ possibilities (because we must choose a subset of a family of $i$ intervals), while for \texttt{(ii)} there are at most $2^{j-2}$ possibilities (because we must choose a subset of all indices $j'$ between $2$ and $j-1$).
Overall, that divides the third constraint in at most $2^i2^{j-2} \leq 2^{2k-2}$ disjoint events. 
Furthermore, for any such event, let us write $(a_1;b_1), (a_2;b_2), \ldots, (a_r;b_r)$ the selected intervals, resp. $j_1,j_2,\ldots,j_{r-1}$ the selected indices. Set $j_0 = 1, \ j_r = j$. Then, for every $1 \leq q \leq r$, all the intervals $[l_{j'}(w);u_{j'}(w)]$, for $j_{q-1} \leq j' \leq j_q-1$ must be contained into $(a_q;b_q)$. We set $[l_{2j'-1};u_{2j'-1}] = [l_{2j'};u_{2j'}] = (a_q;b_q)$ for the $(2j'-1)^{\text{th}}$ and ${2j'}^{\text{th}}$ intervals of the range query.

For a fixed pair $(i,j)$ we so reduce our computation to at most $3 \cdot 2^{2k-2}$ range queries {(obtained by the combination of one of the three disjoint possibilities for the two first constraints with one of the at most $2^{2k-2}$ possibilities for the third constraint)}, that takes time ${\cal O}(2^{{\cal O}(k)}n^{\epsilon})$ for any $\epsilon > 0$ by Proposition~\ref{prop:complexity-range-trees}.
Since there are ${\cal O}(k^2)$ such pairs, the total time in order to compute the number of vertices in $S$ at distance two from $v$ is in ${\cal O}(k^22^{{\cal O}(k)}n^{\epsilon})$ for any $\epsilon > 0$.

Finally, the hardness result for $k = \omega(\log{n})$ follows from the fact that $k$-treewidth split graphs are $k$-clique-interval (Lemma~\ref{lem:bounded-tw}) and that under SETH, we cannot compute the diameter of split graphs with treewidth $\omega(\log{n})$ in subquadratic time~\cite{AVW16,BCH16}.
\end{proofof}

Before ending this section, we give a simpler algorithm for computing the diameter on the {\em complements} of $k$-clique-interval split graphs.
It is similar in spirit to~\cite[Lemma 6]{DHV20}.

\begin{theorem}\label{thm:main-complement}
If $G=(K \cup S,E)$ is { an $m$-edge split graph, and we are given a total ordering over $S$ showing that $G$ is the complement of a $k$-clique-interval split graph} then, we can compute the diameter of $G$ in time ${\cal O}(km)$.
\end{theorem}

\begin{proof}
By the hypothesis for every $v \in K$, we have $N_{\overline{G}}(v) \cap S$ is the union of $k$ intervals.
In particular, $N_{G}(v) \cap S = \bigcup_{i=1}^{k+1} [l_i(v);u_i(v)]$ is the union of $k+1$ intervals such that $l_1(v) \leq u_1(v) \leq l_2(v) \leq u_2(v) \leq \ldots \leq l_{k+1}(v) \leq u_{k+1}(v)$.
Furthermore since the ordering over $S$ is given, the $2k+2$ endpoints that delimit these intervals can be computed in time ${\cal O}(|N_G(v)|)$, simply by scanning the neighbours of vertex $v$. 
Overall, this pre-processing phase takes total time ${\cal O}(m)$.
Then in order to compute the diameter of $G$, for every $w \in S$ we store the $2k+2$ endpoints $l_1(v),u_1(v),l_2(v),u_2(v), \ldots,l_{k+1}(v),u_{k+1}(v)$ for every $v \in N_G(w)$.
This takes time ${\cal O}(k|N_G(w)|)$, and so total time ${\cal O}(km)$.
Furthermore a vertex $w \in S$ has eccentricity at most two if and only if we have $\bigcup_{v \in N_G(w)} N_G[v] = V$, that is equivalent to having $\bigcup_{v \in N_G(w)}\bigcup_{i=1}^{k+1} [l_i(v);u_i(v)] = S$. 
In order to check whether this collection of intervals covers all of $S$, it suffices to sort the pairs $(l_i(v),u_i(v))$ for all $v \in N_G(w)$ and $1 \leq i \leq k+1$, and then to scan these ordered pairs from left to right.
This can be done in time ${\cal O}(k|N_G(w)|)$.
As a result, we can decide whether $diam(G) \leq 2$ in total time ${\cal O}(km)$.
\end{proof}

\section{Recognition of $k$-Clique-interval Split graphs}\label{sec:recognition}

Our two algorithms in Sec.~\ref{sec:gal-case} show that in order to compute the diameter of $k$-clique-interval split graphs in quasi linear time, it is sufficient to compute a corresponding total order of their maximal clique.
This raises the question whether such $k$-clique-interval orderings can always be computed in quasi linear time.
A first positive example was given by Lemma~\ref{lem:bounded-tw}.
Indeed, for a split graph of treewidth at most $k$, we can pick {\em any} total order of its maximal clique.
We complete this easy result by Sec.~\ref{sec:dense} where we give examples of dense subclasses of split graphs with constant clique-interval number and for which a corresponding order can be computed in linear time.
Finally, in Sec.~\ref{sec:main-recognition} we prove a stronger result, namely that we can recognize the clique-interval { split} graphs in linear time.

\subsection{Examples of subclasses with bounded clique-interval number}\label{sec:dense}

A {\em threshold graph} is a split graph $G=(K \cup S,E)$ such that: \texttt{(i)} the neighbourhoods of the vertices in $K$ and \texttt{(ii)} the neighbourhoods of the vertices in $S$ are totally ordered by inclusion.
Observe that threshold graphs can be dense and of unbounded treewidth.

\begin{lemma}\label{lem:threshold}
Every threshold graph is { a clique-interval split graph}.
\end{lemma}

\begin{proof}
For a threshold graph $G=(K \cup S,E)$ let $S=(v_1,v_2,\ldots,v_p)$ be such that $N_G(v_1) \subseteq N_G(v_2) \subseteq \ldots \subseteq N_G(v_p)$.
In order to prove that $G$ is { a clique-interval split graph}, it suffices to construct any total order of $K$ such that the subsets $N_G(v_1), N_G(v_2) \setminus N_G(v_1), \ldots, N_G(v_i) \setminus N_G(v_{i-1}), \ldots, N_G(v_p) \setminus N_G(v_{p-1}), K \setminus N_G(v_p)$ are consecutive intervals. 
\end{proof}

Note that we can easily derive from the proof of Lemma~\ref{lem:threshold} a linear-time algorithm for computing a clique-interval ordering.

Finally, a {\em comparability graph} is a graph that admits a transitive orientation.
{Let $K=(u_1,u_2,\ldots,u_{|K|})$ be totally ordered. In what follows, a prefix of this order (resp., a suffix) is any subset of consecutive vertices $u_1,u_2,\ldots,u_i$ (resp., $u_i,u_{i+1},\ldots,u_{|K|}$) for some $i$.}

\begin{lemma}\label{lem:comparability}
For every comparability split graph $G=(K \cup S,E)$, we can compute in linear time a total order over $K$ such that, for every $v \in S$, $N_G(v)$ is the union of a prefix and a suffix of this order.

In particular, every comparability split graph is $2$-clique-interval.
\end{lemma} 

\begin{proof}
A {\em comparability ordering} of $G$ is a total order $\prec$ over $V = K \cup S$ with the property that, for every $u \prec v \prec w$, $uv,vw \in E \Longrightarrow uw \in E$~\cite{KrS93}.
For a given comparability graph $G$, we can compute a comparability ordering $\prec$ in linear time~\cite{McS97}.
Then, let $\prec_K$ be the subordering induced by $\prec$ over $K$.
For every $v \in S$, we claim that $N_L(v) = \{ w \in K \mid w \prec v \}$ is a (possibly empty) prefix of $\prec_K$.
Indeed, suppose by contradiction there exist {$w,w' \in K$,} $w \in N_L(v), w' \notin N_L(v)$ such that $w ' \prec w \prec v$.
Since $w'w, wv \in E$ we should have $w'v \in E$, a contradiction.
Therefore, the claim is proved.
We can prove similarly that $N_R(v) = \{ w \in K \mid v \prec w \}$ is a (possibly empty) suffix of $\prec_K$. 
\end{proof}

We also want to stress that the complements of comparability split graphs, i.e., the cocomparability split graphs are just interval split graphs and we have already considered this case in Section~\ref{sec:univ-vertex}.

{
\begin{corollary}
The diameter of comparability split graphs can be computed in linear time.
\end{corollary}
\begin{proof}
We first apply Lemma~\ref{lem:comparability}, of which we reuse in what follows the notations that were locally defined in its proof. 
For every $x,y \in S$ such that $N_L(x)$ and $N_L(y)$ are non-empty we have either $N_L(x) \subseteq N_L(y)$ or $N_L(y) \subseteq N_L(x)$ because these are prefixes of the ordering over $K$.
In particular, $x$ and $y$ are at distance two from each other.
The same holds if both $N_R(x)$ and $N_R(y)$ are non-empty.
Hence, let $S_L = \{ x \in S \mid N_R(x) = \emptyset \}$ and $S_R = \{ y \in S \mid N_L(y) = \emptyset \}$.
We are left deciding whether every $x \in S_L, y \in S_R$ are pairwise at distance two.
For that, since all the sets $N_L(x)$ are prefixes of the total order over $K$, and in the same way all the sets $N_R(y)$ are suffixes of this ordering, we only need to consider a pair $x,y$ such that $|N_L(x)|$ and $|N_R(y)|$ are minimized.
\end{proof}
}

\subsection{Linear-time recognition of Clique-interval { split} graphs}\label{sec:main-recognition}

\begin{theorem}\label{thm:clique-interval-recognition}
Clique-interval split graphs can be recognized in linear time.
\end{theorem}

\begin{proof}
Let $G=(K \cup S,E)$ be a split graph.
We define a graph $G^+$ from $G$ by first transforming $K$ into a stable set and then, for every $u \in K$, making a clique of $N_G(u) \cap S$.
Furthermore, we claim that $G$ is { a clique-interval split graph} if and only if $G^+$ is interval.
To see that, let us call {\em clique-path} of a graph $G'$ an ordering of its maximal cliques such that, for any vertex $v$ of $G'$, the maximal cliques containing $v$ are contiguous in the ordering; a graph is interval if and only if it admits a clique-path~\cite{BlairPeyton93,FuG65}.
We can now prove our claim, as follows:
\begin{itemize}
\item If $G$ is { a clique-interval split graph} then, any clique-ordering over $K$ for $G$ is a total ordering over the maximal cliques $\{u\} \cup \left( N_G(u) \cap S \right), \ u \in K$, for $G^+$.
Furthermore as already observed in the proof of Proposition~\ref{prop:clique-interval} the vertices in a subset $S' \subseteq S$ are pairwise at distance two in $G$ if and only if they have a common neighbour in $K$.
As a result, the maximal cliques of $G^+$ are exactly the sets $\{u\} \cup \left( N_G(u) \cap S \right), \ u \in K$, and so $G^+$ admits a clique-path.
This implies that $G^+$ is an interval graph.
\item Conversely, if $G^+$ is an interval graph then any clique-path of $G^+$ induces a total ordering over the maximal cliques $\{u\} \cup \left( N_G(u) \cap S \right), \ u \in K${. Since every vertex $u \in K$ is contained in exactly one such maximal clique, we get a total ordering $\tau$ over $K$. For every $v \in S$, we have that $N_G(v)$ consists of an interval of vertices in $\tau$ because the maximal cliques $\{u\} \cup \left( N_G(u) \cap S \right), \ u \in N_G(v)$ must be consecutive in the clique-path of $G^+$. Therefore, $\tau$ is a clique-ordering over $K$ for $G$.} 
\end{itemize}

Unfortunately, computing the graph $G^+$ may take super-linear time.
We can overcome this issue as follows.
First, given a family ${\cal C}$ of subsets over $V$, if there exists a chordal graph $G_{{\cal C}}$ whose maximal cliques are exactly those in ${\cal C}$ then, we can compute a Lex-BFS ordering of $G_{{\cal C}}$ in time ${\cal O}(\sum_{C \in {\cal C}}|C|)$ (Algorithm 10 in~\cite{HMPV00}).
Moreover we can also compute a clique-tree of $G_{{\cal C}}$ in time ${\cal O}(\sum_{C \in {\cal C}}|C|)$~\cite[Sec. 3]{TaY84}.
Finally given this Lex-BFS ordering and the corresponding clique-tree of $G_{{\cal C}}$, we can apply Algorithm 9 from~\cite{HMPV00} in order to decide in time ${\cal O}(\sum_{C \in {\cal C}}|C|)$ whether $G_{\cal C}$ is interval.
For solving our initial problem, we can take ${\cal C} = \{ \{u\} \cup \left( N_G(u) \cap S \right) \mid u \in K \}$, and in this situation we have $\sum_{C \in {\cal C}}|C| = {\cal O}(n+m)$. 
\end{proof}

We left open the status of the recognition of $k$-clique-interval split graphs, for $ k \geq 2$.

\section{Open problems}\label{sec:open-pb}

Although the definitions of $k$-clique-interval and $k$-interval split graphs have some similarities, we observe that computing the diameter of $2$-interval split graphs in quasi linear time already looks like a challenging task.
Indeed, for a $2$-interval split graph $G=(K \cup S,E)$ and $v \in S$, the vertices $u \in S$ at distance two from $v$ are exactly those such that one of their $2$ intervals intersects one of the $2|N_G(v)|$ intervals that represent the neighbours of $v$.
We cannot use our range query framework in order to avoid overcounting these vertices as this would require up to $2^{{\cal O}(|N_G(v)|)}$ range queries. 
More generally, for every fixed $k > 1$, {\em can we compute the diameter of $k$-interval graphs in quasi linear time?}
We stress that every planar graph is $3$-interval~\cite{ScW83}, and that the complexity of diameter computation on this class of graphs is a longstanding open problem. The case $k=2$ could thus be an interesting intermediate step.

\nocite{*}
\bibliographystyle{abbrvnat}
\bibliography{split-biblio}
\label{sec:biblio}

\end{document}